\documentclass[11pt]{article}

\usepackage[margin=1.15in]{geometry}
\usepackage{amsmath,amssymb,amsthm}
\usepackage[hidelinks]{hyperref}

\hypersetup{
  pdftitle={Safe Quotes for Retroactive Liquidity Pools},
  pdfauthor={Peter Bro Miltersen}
}

\newtheorem{theorem}{Theorem}
\newtheorem{lemma}{Lemma}

\newcommand{\init}{\textnormal{\textbf{init}}}
\newcommand{\provide}{\textnormal{\textbf{provide}}}
\newcommand{\reclaim}{\textnormal{\textbf{reclaim}}}
\newcommand{\provides}{\provide{}s}
\newcommand{\reclaims}{\reclaim{}s}
\newcommand{\lockSwapAtoB}{\textnormal{\textbf{lockSwapAtoB}}}
\newcommand{\lockSwapBtoA}{\textnormal{\textbf{lockSwapBtoA}}}
\newcommand{\execute}{\textnormal{\textbf{execute}}}
\newcommand{\cancel}{\textnormal{\textbf{cancel}}}

\title{Safe Quotes for Retroactive Liquidity Pools}
\author{Peter Bro Miltersen\footnote{Peter Bro Miltersen is an independent researcher. He dedicates this paper to the memory of Kurt Nielsen who first suggested the lock-swap mechanism and worked hard to get it deployed. Kurt sadly passed away much too early in July 2026.}}
\date{July 30, 2026}

\begin{document}
\maketitle

\begin{abstract}
Automated market makers exchange assets through liquidity pools whose quoted
prices depend on their reserves, with constant product pools being the most common.  When such pools reside on different blockchains
or shards, a sequence of swaps cannot in general be executed atomically.
Aanes {\em et al.} introduced lock-swaps and
retroactive constant product liquidity pools to provide price guarantees
for such a setting. A
retroactive pool implicitly maintains a virtual pool for each possible execute/cancel resolution of its active locks. In the presence of active locks, serving a new swap request requires computing a {\em safe quote}; a quote with an output
that does not exceed the minimum possible output, taken over all virtual pools. The quote being safe is a hard constraint ensuring the integrity of the pool. A soft constraint is to make the quote as close to the minimum possible output as possible. Aanes {\em et al.} gave a
simple and efficient algorithm for computing the exact minimum when
unresolved \provides{} and \reclaims{} of liquidity do
not coexist, showed by an explicit example that the algorithm fails in
general, and left the computational complexity of the general case open.
In this paper, we show that unless P is equal to NP, there is no polynomial time algorithm that computes in the general case a safe quote with any fixed multiplicative approximation ratio (e.g., 50\%) relative to the exact minimum. This seems like a severe obstacle for deployment of the lock-swap functionality. However, we also present two simple and practical algorithms for computing safe quotes that have input-dependent approximation ratios that are likely to be satisfactory in practice, thus circumventing that obstacle.
\end{abstract}

\noindent\textit{Keywords:} DeFi, Automated Market Makers, Constant Product
Liquidity Pools, Sharded Blockchains, Cross-chain DeFi.

\section{Introduction}
\label{sec:introduction}

Constant-product automated market makers allow traders to exchange two assets
without a conventional order book.  Such a market maker is particularly easy
to use when all swaps in a transaction are executed atomically: an
arbitrageur can submit a cycle of swaps and arrange that either the entire
cycle is executed or none of it is.  This protection is generally unavailable
when the pools involved lie on different blockchains or on different shards
of a sharded blockchain.  While a trader waits for one swap to settle,
another trade may change the price offered by a later pool, and a planned
arbitrage cycle may end in a loss.

Aanes, Gravgaard, Miltersen, Nielsen, and Pourpouneh
\cite{aanes2025lockswap}
 introduced
\emph{lock-swaps} to recover the relevant price guarantee without freezing
the pool.
A lock-swap gives its holder an option: acquire a quote for a swap,
execute the swap later, or cancel it.  A trader can first acquire one
such lock for each leg of a cross-chain transaction and then execute all the locked swaps
if the complete chain of transactions is satisfactory.  Other traders may continue to
swap in the meantime, and liquidity providers may continue to provide and
reclaim liquidity; the locks fix the quoted prices, not the pool itself.
Aanes {\em et al.} implement the lock-swap functionality by a construction called a \emph{retroactive pool}. An
active lock must remain executable whichever of the already-active locks are
later executed or canceled, so the
retroactive pool conceptually maintains one
\emph{virtual pool} for each execute/cancel resolution of the active locks;
with $k$ active locks there are $2^k$ virtual pools.  The default quote for
a new
lock-swap is then the minimum, over all virtual pools, of the exact swap output.  The
computational challenge is to find this minimum without explicitly enumerating the $2^k$ virtual pools, incurring exponential time complexity.
For an important restricted case, Aanes {\em et al.} overcome that challenge
completely.  Their Theorem~4 states that if the current state contains no
unresolved \provide{} operations, or no unresolved \reclaim{} operations,
then the minimizing resolution is fixed and simple: for an A-to-B lock-swap
request, execute every active A-to-B lock and cancel every active B-to-A lock.
Only one virtual pool per swap direction is then needed, a quote takes time
linear in the stored state, and if \reclaims{} are disallowed while locks are
active, two running sums give constant-time quotes.  The paper also shows, by
an explicit counterexample, that the simple rule fails as soon as unresolved
\provides{} and \reclaims{} coexist.  
Aanes {\em et al.} accordingly posed an open problem: can the general minimum
be computed in polynomial time, is it NP-hard, and if it is hard, do useful
approximation algorithms exist?

The present paper resolves the open problem in a negative and a positive
direction.  We establish NP-hardness of computing not only the exact minimum
quote among the virtual pools,  but also \emph{safe} quotes with any constant
multiplicative approximation ratio (say, 50\%) relative to that minimum.  Here, a quote is safe if it is less than or equal to the exact minimum. Being safe is the hard constraint needed to ensure the integrity of the pool: a conservative quote may be less
attractive to a trader, but a too optimistic quote can leave the reference
pool unable to honor the quote of another active lock or even bring about a
violation of the constant product invariant of the pool. The NP-hardness of
computing such quotes with any constant-factor approximation guarantees seems
a serious obstacle for the deployment of retroactive pools. Fortunately, we
are able to give two algorithms that are computationally efficient in theory as well as in practice, with non-constant approximation ratios
that are functions of a directly observable load parameter.  When the assets of the stored event
list are light relative to the base pool (as is likely to be the case in
practice), both algorithms lead to quotes that are provably close to optimal. We therefore find it likely that these algorithms could be used for enabling a practical deployment of retroactive liquidity pools. 

We recall the model and fix notation before stating and proving our results.  The
presentation below is sufficient for the proofs in this paper; we refer to Aanes {\em et al.} for the description of
the full abstract data type that the retroactive pool implements, its correctness
properties, and discussions of the cross-chain protocol using it.

\subsection{Constant-product pools and liquidity operations}
\label{sec:pool-model}

A pool holds positive amounts $a$ and $b$ of two assets, called A and B.  We
call $(a,b)$ its \emph{asset amount pair}, or its \emph{allocation}.  In the
fee-free, or \emph{exact}, model, a trader who sends $x>0$ units of A receives
$y>0$ units of B determined by
\[
  (a+x)(b-y)=ab.
\]
Thus
\begin{equation}
  y=\frac{bx}{a+x}.
  \label{eq:exact-output}
\end{equation}
The product $ab$ is unchanged by an exact swap.  Swaps with fees can only
increase it, but the open problem studied here assumes exact swaps in order to
make the quote calculation unambiguous.

Liquidity providers own shares of the pool represented by liquidity tokens.
A
\emph{liquidity token portion} is a real-valued amount of such tokens.
We write $z>0$ for the total amount of liquidity tokens minted and not yet
burned.  The initialization operation
\[
  t\leftarrow \init(a,b)
\]
creates the asset amount pair $(a,b)$, sets $z=1$, and returns a token portion
$t$ of amount $1$.

There are two operations for changing the supplied liquidity.  A
\emph{provide} operation
\[
  t\leftarrow \provide(p,q), \qquad p,q\geq 0,\quad p+q>0,
\]
adds $p$ units of A and $q$ units of B.  The amounts need not be proportional
to the current pool holdings.  Starting from $(a,b,z)$, the operation produces
\begin{equation}
  a'=a+p,\qquad b'=b+q,\qquad
  z'=z\sqrt{\frac{(a+p)(b+q)}{ab}},
  \label{eq:provide-replay}
\end{equation}
and returns a new liquidity token portion of amount $z'-z$.  This token rule
is the non-proportional-provide rule of~\cite{aanes2025lockswap}; it agrees
with the usual proportional rule when $p/a=q/b$ and can in general in the absense of fees and locks be simulated by first doing a swap making the provider's assets proportional and then a provide by the
proportional rule.

A \reclaim{} operation returns a previously issued token portion to the
pool.  If the portion has amount $r<z$, it is burned and the provider receives
the same fraction $r/z$ of each asset.  Equivalently,
\begin{equation}
  (a,b,z)\longmapsto
  \left(\left(1-\frac rz\right)a,
        \left(1-\frac rz\right)b,z-r\right).
  \label{eq:reclaim-replay}
\end{equation}
The condition $r<z$ leaves both asset amounts positive.  In operational terms,
\provide{} deposits assets and creates a claim on the pool, while \reclaim{}
burns such a claim and withdraws the corresponding share of the assets.

\subsection{Locks and virtual pools}
\label{sec:virtual-pool-model}

A \emph{lock} is a signed asset delta $(\alpha,\beta)$ with one positive and
one negative coordinate.  An A-to-B lock has $\alpha>0$ and $\beta<0$: executing
it sends $\alpha$ units of A to the pool and removes $-\beta$ units of B.
Canceling it leaves the pool unchanged.  A lock is \emph{active}, or
\emph{unresolved}, from the time it is granted until its holder executes or
cancels it.
We write
\[
  L\leftarrow\lockSwapAtoB(x)
  \qquad\text{or}\qquad
  L\leftarrow\lockSwapBtoA(y)
\]
for a request for an A-to-B or B-to-A lock, respectively.  The holder later
resolves $L$ by performing \execute$(L)$ or \cancel$(L)$.

The central difficulty is that later operations must be valid whichever of
the active locks are eventually executed.  A retroactive pool therefore
maintains one \emph{virtual pool} for each execute/cancel assignment to the
active locks.  If there are $k$ active locks, there are $2^k$ assignments.
Each virtual pool is the state obtained by replaying the common history with
the selected lock resolutions.  When a new A-to-B lock with input $x$ is
requested, the retroactive pool evaluates the exact output in
Equation~\eqref{eq:exact-output} for every virtual pool and grants the smallest
of these outputs.  The lock is consequently executable under every later
resolution of the earlier locks.
A \emph{trace} is a chronological record of retroactive-pool operations and
their returned values.  It is \emph{valid} if every operation and return value
obeys the preceding rules: in particular, each lock receives the minimum
output over the current virtual pools, no token portion is reclaimed twice,
and every \reclaim{} is valid in every current virtual pool.

The full histories need not be stored.  The baseline implementation
of~\cite{aanes2025lockswap} uses a compressed representation consisting of a
base state $(a_0,b_0,z_0)$ and a chronological event list beginning with the
earliest active lock.  The entries have four forms:
\begin{enumerate}
  \item an active lock, represented by its signed delta $(\alpha,\beta)$;
  \item a \provide{}$(p,q)$ whose token portion cannot yet be finalized;
  \item a \reclaim{} of a token portion of amount $r$ whose asset withdrawal
        cannot yet be finalized; or
  \item the aggregate signed delta $(\alpha,\beta)$ of adjacent locks that
        have already been resolved.
\end{enumerate}
A \provide{} or \reclaim{} is called \emph{unresolved} while a lock that was active
when the operation was requested remains unresolved.  The assets supplied by
a \provide{} can be put to work immediately, but its token amount depends on the
virtual pool in which~\eqref{eq:provide-replay} is evaluated.  Similarly,
Equation~\eqref{eq:reclaim-replay} gives different withdrawals for different
virtual token supplies.  The token portion minted by an unresolved \provide{}
is not returned to its provider until the preceding locks have been resolved
and therefore cannot be named by a \reclaim{} in the current event list.  An
unresolved \reclaim{}, by contrast, immediately surrenders its named token
portion: the portion cannot be named again, and every virtual replay subtracts
its amount from $z$ at the position of the request.  Only the asset withdrawal
and its incorporation into the compressed base state await resolution of the
preceding locks.  Every unresolved \reclaim{} therefore burns a distinct
portion already included in the base token supply $z_0$.

Given an execute/cancel assignment, the corresponding virtual pool is computed
by scanning the event list from $(a_0,b_0,z_0)$.  An executed lock adds
$(\alpha,\beta)$ and a canceled lock does nothing.  A resolved-lock aggregate
always adds its delta.  A \provide{} and a \reclaim{} use the transitions in
Equations~\eqref{eq:provide-replay} and~\eqref{eq:reclaim-replay}, respectively.
We call a compressed state \emph{reachable} if it is produced by a valid
sequence of retroactive-pool operations.  

\subsection{The quote problem}
\label{sec:quote-problem}

Fix a reachable state and an exact A-to-B lock-swap request with input $x>0$.
For a resolution $\tau$ of the active locks, let $(a_\tau,b_\tau)$ be the final
asset amount pair obtained in the corresponding virtual pool.
By~\eqref{eq:exact-output}, the output in this virtual pool is
\[
  x\frac{b_\tau}{a_\tau+x}.
\]
The worst exact quote is therefore
\begin{equation}
  \mu=\min_\tau\frac{x b_\tau}{a_\tau+x}.
  \label{eq:minimum-quote}
\end{equation}
Direct enumeration takes time exponential in the number of active locks.

For deployment, an underestimate of~\eqref{eq:minimum-quote} is more useful than
an ordinary feasible solution to this minimization problem.  We call
$\ell\geq0$ a \emph{safe quote} if $\ell\leq\mu$.  Both $\mu$ and $\ell$ are
amounts of asset B.  In the real-valued replay model, a safe quote does not
exceed the output of any current virtual pool.

For $\rho\geq1$, a safe quote is a \emph{$\rho$-approximation} if
\[
  \frac{\mu}{\rho}\leq\ell\leq\mu.
\]
This is the
reverse of the usual approximation direction for a minimization problem: a
feasible lock resolution gives an upper bound on $\mu$, whereas a safe quote
must be a lower bound.

\subsection{Results and organization}

Section~\ref{sec:hardness} closes the road to exact quotes.  It proves that
the exact quote problem is NP-hard.  A second construction quotes a later
lock in a known minimizing resolution and then removes that resolution by
canceling an earlier lock.  This amplifies the reduction's additive gap and rules
out every polynomial time computation of safe quotes within any constant factor of the exact minimum quote unless
$\mathrm{P}=\mathrm{NP}$.  
Section~\ref{sec:safe-quotes} gives algorithms for computing safe quotes.  The first maintains upper bounds on the two reserves and a lower
bound on the liquidity-token supply.  The second is an inexpensive
input-output balance quote.  Both algorithms scan the event list once, use a
fixed number of arithmetic registers, avoid square roots, and return safe
quote values.  The first is strictly positive on every reachable state.
The NP-hardness result rules out a universal constant approximation ratio, but both bounds have explicit approximation
guarantees as a function of cumulative relative load.  Our measure for that, $\eta$, adds
the total \reclaim{}-token fraction to the larger of the two cumulative relative
asset loads.  It is dimensionless, is computed from the stored state in one
pass, and directly expresses whether the operations performed after the earliest active lock-swap are light relative
to the pool.  For $\eta<1$, the product and balance quotes have
approximation factors at most
$((1+\eta)/(1-\eta))^3$ and $((1+\eta)/(1-\eta))^2$, respectively.  At total
relative load $\eta=0.001$, the guarantees evaluate to at least $99.40\%$ and
$99.60\%$ of the exact minimum quote, respectively.
Section~\ref{sec:conclusion} gives concluding remarks.

\section{Hardness of Exact and Approximate Quotes}
\label{sec:hardness}

We first show NP-hardness of the following exact minimization problem, which
captures the problem of determining the best safe quote.  We then amplify the
gap of the reduction to obtain a non-approximability result.

\medskip
\noindent\textsc{Exact Virtual Quote}:
Given a valid finite trace of retroactive-pool operations whose numerical data
are rational, a positive rational $x$, and a rational $\theta$, decide whether,
at the end of the trace, some virtual pool has an allocation $(a,b)$ for which
\(
  \frac{xb}{a+x}\leq\theta.
\)
The promise is that the input trace is valid.

\begin{theorem}
\label{thm:exact-hardness}
{\normalfont\textsc{Exact Virtual Quote}} is an NP-hard promise problem.
\end{theorem}

The source problem is positive-integer \textsc{Subset Sum}: given positive integers
$w_1,\ldots,w_n$ and a positive integer target $T$, decide whether there is a
set $I\subseteq\{1,\ldots,n\}$ such that
$\sum_{i\in I}w_i=T$~\cite{gareyjohnson1979}.
The cases outside $0<T<\sum_iw_i$ are decided immediately and can be mapped
to fixed yes- and no-instances, so this restricted source problem remains
NP-hard.

\begin{lemma}
\label{lem:reachable-encoding}
Let $w_1,\ldots,w_n$ and $T$ be positive integers satisfying
$0<T<S:=\sum_iw_i$, and set
\[
  L=16S,\qquad K=L+T,\qquad Q=3L.
\]
One can construct, in polynomial time, a valid trace that begins with
\[
  \init(1/5,K/5),\qquad
  t^\star\leftarrow\provide(4/5,4K/5).
\]
After these operations the allocation is $(1,K)$, the token supply is $5$,
and $t^\star$ has amount $4$.  For
\[
  P_i=\sum_{j<i}w_j,\qquad
  B_i=K-P_i,\qquad
  \delta_i=\frac{K w_i}{B_i(B_i-w_i)},
\]
the trace next requests one active A-to-B lock $X_i$ with delta
$(\delta_i,-w_i)$ for each $i\in\{1,\ldots,n\}$.  The trace ends with
an unresolved \provide$(0,Q)$ and an unresolved \reclaim$(t^\star)$.
For every subset $I\subseteq\{1,\ldots,n\}$, there is a resolution that
executes precisely the locks $X_i$ with $i\in I$.  If
\[
  W=\sum_{i\in I}w_i,
  \qquad
  B=L+T-W,
\]
then the resulting final allocation $(a,b)$ has $0<a<2$ and
\begin{equation}
  b=f(B):=B+3L-\frac45\sqrt{B(B+3L)}.
  \label{eq:target-function}
\end{equation}
\end{lemma}

\begin{proof}
Initialize a retroactive pool by
\[
  \init(1/5,K/5)
\]
and then perform the proportional addition
\[
  \provide(4/5,4K/5).
\]
The second operation mints a liquidity token portion $t^\star$ of amount $4$.
The resulting asset amount pair is $(1,K)$, and the number of minted and
unburnt liquidity tokens is $5$.

When $X_i$ is requested, all preceding active locks are A-to-B, and
there is no unresolved \provide{} or \reclaim{}.  Theorem~4
of~\cite{aanes2025lockswap} states that, under this condition, the minimum
output for a new A-to-B lock is attained by the resolution that executes every
active A-to-B lock and cancels every active B-to-A lock.  Hence, in the present
construction, the minimum is attained by executing every lock $X_j$ with
$j<i$.
By induction, the allocation under this resolution is
\[
  \left(\frac K{B_i},B_i\right).
\]
Request
\[
  X_i\leftarrow\lockSwapAtoB(\delta_i).
\]
The exact output under the minimizing resolution is
\[
  \frac{B_i\delta_i}{K/B_i+\delta_i}=w_i.
\]
The lock returned by the operation therefore has asset delta
$(\delta_i,-w_i)$.

After requesting the locks $X_1,\ldots,X_n$, perform
\[
  \provide(0,Q)
  \qquad\text{and}\qquad
  \reclaim(t^\star).
\]
The portion $t^\star$ was issued before the locks $X_i$, and its amount is
$4<5$ when the \reclaim{} is requested.  Both operations are valid.  Since the
locks $X_i$ remain active, the \provide{} and \reclaim{} are unresolved at the
end of the trace and are replayed in every virtual pool.  Thus every subset
$I$ occurs by executing exactly the locks $X_i$ with $i\in I$.

Fix an execute/cancel assignment and let
\[
  W=\sum_{i\text{ executed}}w_i.
\]
Immediately before the B-only \provide{}, the virtual B amount is
\[
  B=K-W=L+T-W.
\]
The corresponding A amount is positive and uniformly bounded.  Indeed, the
all-executed increments telescope, and every assignment satisfies
\[
  a\leq \frac K{K-S}\leq\frac{17}{15}<2.
\]

The B-only \provide{} multiplies the number of liquidity tokens by
\[
  \sqrt{\frac{B+Q}{B}}.
\]
Burning $t^\star$ then scales both asset amounts by
\[
  1-\frac45\sqrt{\frac{B}{B+Q}}.
\]
Consequently the final B amount is $f(B)$.  All event parameters are rational
and have polynomial bit length, so the construction runs in polynomial time.
\end{proof}

\begin{lemma}
\label{lem:subset-sum-gap}
With $S,L$, and $f$ as in Lemma~\ref{lem:reachable-encoding}, $f$ is uniquely
minimized at $B=L$.  For every integer
$B\in[L-S,L+S]$ with $B\ne L$,
\[
  f(B)-f(L)>\frac1{100L}.
\]
\end{lemma}

\begin{proof}
Differentiation gives
\[
  f'(B)=1-\frac45\frac{2B+3L}{2\sqrt{B(B+3L)}}
\]
and
\[
  f''(B)=\frac45\frac{(3L)^2}{4(B(B+3L))^{3/2}}>0.
\]
Since $f'(L)=0$, strict convexity makes $L$ the unique minimizer, with
\[
  f(L)=\frac{12L}{5}.
\]
The relevant interval lies in $(0,2L]$, so $B+3L\leq5L$ and
$(B(B+3L))^{3/2}<32L^3$.  Hence
\[
  f''(B)>\frac{9}{160L}
\]
throughout the interval.  Taylor's theorem and $|B-L|\geq1$ now give
\[
  f(B)-f(L)
  \geq\frac12\frac9{160L}(B-L)^2
  >\frac1{100L}.
\]
\end{proof}

\begin{proof}[Proof of Theorem~\ref{thm:exact-hardness}]
Apply Lemma~\ref{lem:reachable-encoding} to the given \textsc{Subset Sum}
instance.  Set
\[
  \Delta=\frac1{100L},
  \qquad
  x=10^6L^2,
  \qquad
  \theta=\frac{12L}{5}+\frac{\Delta}{2}.
\]
If some subset sums to $T$, its lock assignment has $B=L$.  By
Lemmas~\ref{lem:reachable-encoding} and~\ref{lem:subset-sum-gap}, its final B
amount is $12L/5$, and its quote value is
\[
  \frac{x(12L/5)}{a+x}<\frac{12L}{5}<\theta.
\]

Suppose that no subset sums to $T$.  Every assignment then has $B\ne L$, so
Lemmas~\ref{lem:reachable-encoding} and~\ref{lem:subset-sum-gap} give a final
B amount of at least $12L/5+\Delta$ and a final A amount less than $2$.  Its
quote value is therefore at least
\[
  \frac{x(12L/5+\Delta)}{x+2}.
\]
This quantity exceeds $\theta$.  It suffices to verify
\[
  \frac{x\Delta}{2}
  >2\left(\frac{12L}{5}+\frac\Delta2\right).
\]
The left side is $5000L$, while the right side is below $5L$ for $L\geq1$.

The quote instance is consequently a yes-instance exactly when the original
\textsc{Subset Sum} instance is a yes-instance.
Lemma~\ref{lem:reachable-encoding} and the definitions of $x$ and $\theta$
show that the reduction runs in polynomial time.
\end{proof}

The traces constructed above and their compressed endpoints both
have polynomial size.  Hence the hardness result also holds when the input is
a compressed state promised to be reachable, as in an implementation.

The additive separation in the preceding reduction is only $\Theta(1/L)$
against a quote-value scale of $\Theta(L)$.  We now show how a later exact
lock can turn this small additive separation into an arbitrarily large
multiplicative one.  The construction uses a known minimizing resolution to
determine that lock's output and then cancels the earlier lock needed for that
resolution.

\begin{theorem}
\label{thm:no-constant-factor}
Fix $\rho>1$.  If there is a polynomial-time algorithm that takes as input a
valid trace whose numerical data are rational and a new exact lock-swap
request to be made at the end of that trace, and always returns a
$\rho$-approximate safe quote for that request, then
$\mathrm{P}=\mathrm{NP}$.
\end{theorem}

\begin{proof}
Use the \textsc{Subset Sum} notation and the initial operations stated in
Lemma~\ref{lem:reachable-encoding}.  Thus $L=16S$, $K=L+T$, the current
allocation is $(1,K)$, the token supply is $5$, and the token portion $t^\star$ of
amount $4$ is available.  Put
\[
  D=S-T.
\]
Before requesting any lock $X_i$, request a B-to-A lock with B input $D$.  Its
exact A output in the unique current pool is
\[
  c=\frac{D}{K+D},
\]
so this calibration lock $C$ has delta $(-c,D)$.

Next request the locks $X_i$ with the inputs $\delta_i$ defined in
Lemma~\ref{lem:reachable-encoding}.  No \provide{} or \reclaim{} is unresolved
at this point.  For each new A-to-B quote,
Theorem~4 of~\cite{aanes2025lockswap} identifies the minimizing resolution:
cancel the active B-to-A calibration lock $C$ and execute every A-to-B lock
$X_j$ with $j<i$.  Thus the quote calculations are the same as in
Lemma~\ref{lem:reachable-encoding}, and the returned deltas are again
$(\delta_i,-w_i)$.  Now perform
\[
  \provide(0,3L),\qquad \reclaim(t^\star).
\]
These are the valid unresolved operations from
Lemma~\ref{lem:reachable-encoding}.

Let $\chi\in\{0,1\}$ record whether $C$ executes, and let $W$ be the total
weight of the executed locks $X_i$.  Immediately before the \provide{}, the B
reserve is
\[
  B=K+\chi D-W.
\]
This is an integer in $[L-S,L+S]$, since $0\leq W\leq S$ and
$0<D<S$.
The final B reserve is consequently $f(B)$, where $f$ is defined in
Lemma~\ref{lem:reachable-encoding}.  If $C$ and all locks $X_i$ execute, then
$B=K+D-S=L$.
This known calibration resolution has final reserves
\[
  a_c=\frac35\left(\frac{K}{K-S}-c\right),
  \qquad
  b_* = f(L)=\frac{12L}{5}.
\]
Both values are rational because the \provide{} multiplier at $B=L$ is $2$ and
the \reclaim{} survival factor is $3/5$.

Any other resolution with final B reserve $b_*$ must cancel $C$ and execute
the locks indexed by a set $I$ for which $\sum_{i\in I}w_i=T$.  Its A reserve is
\[
  a_I=\frac35\left(1+\sum_{i\in I}\delta_i\right)<a_c.
\]
To see the strict inequality, observe that
\[
  \delta_i=\frac{K w_i}{B_i(B_i-w_i)}>\frac{w_i}{K}
  \quad\text{and}\quad
  c=\frac{D}{K+D}<\frac{D}{K}.
\]
The indices outside $I$ have total weight $D$, and hence
\[
  \frac{K}{K-S}-c-\left(1+\sum_{i\in I}\delta_i\right)
  =\sum_{i\notin I}\delta_i-c>0.
\]
Every resolution has final A reserve less than $2$.  Every resolution whose B reserve
before the \provide{} differs from $L$ has final B reserve greater than
$b_*+\Delta$, where
\[
  \Delta=\frac1{100L}.
\]

Since $\rho$ is fixed, hard-wire an integer $g\geq\rho$ into the reduction
and set
\[
  H=2000gL^2.
\]
Request an A-to-B lock with input $H$.  The exact output from allocation
$(a,b)$ is
\[
  G_H(a,b)=\frac{bH}{a+H}.
\]
The calibration resolution uniquely minimizes $G_H$.  Among resolutions with
$b=b_*$, it has the largest A reserve, and $G_H(a,b_*)$ decreases with $a$.
Every resolution whose B reserve before the \provide{} differs from $L$ gives
\[
  G_H(a,b)
  >\frac{(b_*+\Delta)H}{H+2}
  >b_*
  >G_H(a_c,b_*),
\]
where the middle inequality follows from
$H\Delta=20gL>2b_*=24L/5$.  The new amplifier lock therefore has the known
rational delta
\[
  (H,-v),
  \qquad
  v=\frac{b_*H}{a_c+H}.
\]

Cancel the earlier calibration lock $C$ while leaving the locks $X_i$ and the
amplifier active.  The retroactive-pool operation removes the virtual pools in
which $C$ executes; it does not change the delta of the later lock.  Since a
lock $X_i$ still precedes the \provide{} and \reclaim{}, these liquidity events
remain unresolved.  The resulting valid trace therefore ends in a state from
which the resolution that determined $v$ is absent.

Request an A-to-B lock with input $x=1$ in this state.  Write
\[
  \varepsilon=b_*-v=\frac{b_*a_c}{a_c+H}.
\]
Every remaining resolution has B reserve at least $b_*$.  If the amplifier is
canceled, its quote value is greater than $b_*/3=4L/5$.  If the amplifier
executes, its quote value is less than $5L/H$, since
$f(B)<B+3L<5L$.  The choice of $H$ makes $5L/H<4L/5$, so every minimizing
resolution executes it.

If some subset sums to $T$, that subset has B reserve $b_*$ before the
amplifier and gives
\[
  \mu<\frac{\varepsilon}{H}<\frac{2b_*}{H^2}=:Y.
\]
If no subset sums to $T$, every resolution has B reserve greater than
$b_*+\Delta$ before the amplifier.  Since its A reserve is less than $2$,
\[
  \mu>\frac{\Delta}{H+3}=:N.
\]
The two rational bounds satisfy
\[
  \frac{N}{Y}
  =\frac{\Delta H^2}{2b_*(H+3)}
  >\frac{\Delta H}{4b_*}
  =\frac{25g}{12}
  >g.
\]

Run the assumed safe approximation and compare its rational output $\ell$
with $Y$.  In a yes-instance, $\ell\leq\mu<Y$.  In a no-instance,
$\ell\geq\mu/\rho>N/\rho>Y$.  The comparison decides \textsc{Subset Sum}.

Every numerical value recorded in the constructed trace is rational with
polynomial binary length.  Square roots occur in non-calibration virtual
replays, but they are neither recorded in the trace nor evaluated by the
reduction.  Lemma~\ref{lem:subset-sum-gap} supplies the required comparison:
if $B\ne L$, then $f(B)>b_*+\Delta$.
Although $H$ is numerically large, its binary representation has
$O(\log g+\log L)$ bits.  Thus the reduction runs in polynomial time in the
standard Turing machine model.
\end{proof}

\section{Efficiently Computable Safe Quotes}
\label{sec:safe-quotes}

We now give the constructive results.  The first quote follows from a product
invariant and three scalar bounds.  The second is an input-output balance
quote that costs two additional accumulators.

Write $(\alpha,\beta)$ for the asset delta of a lock or a resolved-lock
aggregate, $(p,q)$ for the nonnegative amounts in a \provide{}, and $r$ for the
amount of a liquidity token portion burned by a \reclaim{}.  Let $m$ be the
event-list length.  Each quote computation uses $O(m)$ arithmetic
operations and a fixed number of scalar registers.

\subsection{A universal product-floor quote}

For the reachable compressed state given as input, write $(a_0,b_0,z_0)$ for
its base state: $a_0$ and $b_0$ are the base asset amounts, and $z_0$ is the
base amount of minted and unburnt liquidity tokens.
For an allocation $(a,b)$ with $z$ minted and unburnt liquidity tokens, define
\[
  \gamma=\frac{z}{\sqrt{ab}},
  \qquad
  \gamma_0=\frac{z_0}{\sqrt{a_0b_0}}.
\]

\begin{lemma}
\label{lem:invariant}
At every prefix of every virtual replay of a reachable compressed state,
\[
  \gamma\leq\gamma_0.
\]
\end{lemma}

\begin{proof}
A \provide{} $(p,q)$ multiplies both $z$ and $\sqrt{ab}$ by
\[
  \sqrt{(1+p/a)(1+q/b)},
\]
and a \reclaim{} multiplies both by $1-r/z$.  These operations preserve $\gamma$.

Consider an A-to-B lock with delta $(u,-v)$, where $u,v>0$.  When the lock is
granted, its output satisfies, for every virtual allocation to which it may
later be applied,
\[
  v\leq\frac{bu}{a+u}.
\]
Indeed, the prefix of a later replay immediately before this lock is one of
the virtual pools considered when the lock was granted.
Equivalently,
\[
  (a+u)(b-v)\geq ab.
\]
Execution cannot decrease $\sqrt{ab}$, while cancellation leaves it
unchanged; neither operation changes $z$.  The B-to-A case is symmetric.
An aggregate of resolved locks represents a sequence of such executions.
Induction over the event list proves the claim.
\end{proof}

The algorithm processes the given state's entire stored event list in
chronological order while maintaining reserve upper bounds $A^+,B^+$ and a
token lower bound $Z^-$.  Initialize
\[
  A^+=a_0,\qquad B^+=b_0,\qquad Z^-=z_0.
\]
For a delta entry $(\alpha,\beta)$, set
\[
  A^+\leftarrow A^++\max\{0,\alpha\},
  \qquad
  B^+\leftarrow B^++\max\{0,\beta\}.
\]
For a \provide{} $(p,q)$, first compute
\[
  \xi=\frac p{A^+}+\frac q{B^+}+\frac{pq}{A^+B^+},
  \qquad
  F^-=1+\frac{\xi}{2+\xi}.
\]
and then set
\[
  Z^-\leftarrow Z^-F^- ,\qquad
  A^+\leftarrow A^++p,\qquad
  B^+\leftarrow B^++q.
\]
For a \reclaim{} of amount $r$, set $Z^-\leftarrow Z^--r$.

\begin{theorem}
\label{thm:product-quote}
For every reachable compressed state and every exact A-to-B lock-swap request
with input $x>0$, let $A^+,B^+,Z^-$ be the terminal register values produced
by the algorithm above.  Then
\begin{equation}
  \ell_{\mathrm{prod}}
  =
  \frac{x(Z^-/z_0)^2a_0b_0}
       {A^+(A^++x)}
  \label{eq:product-quote}
\end{equation}
is a strictly positive safe quote.
\end{theorem}

\begin{proof}
We first verify that $a\leq A^+$ and $b\leq B^+$ after every prefix of every
virtual-pool replay.  The inequalities hold with equality at initialization.
An active delta entry contributes either $(0,0)$ or $(\alpha,\beta)$, and a
resolved aggregate contributes $(\alpha,\beta)$; in either case, the increase
in each reserve is at most the positive part used in the update above.  A
\provide{} adds exactly $(p,q)$ to both the replayed reserves and their bounds.
A \reclaim{} multiplies both replayed reserves by a number in $(0,1)$ while
leaving $A^+$ and $B^+$ unchanged.  Induction proves the two upper bounds.

At a \provide{}, its exact token
multiplier is
\[
  F=\sqrt{(1+p/a)(1+q/b)}.
\]
Since $a\leq A^+$ and $b\leq B^+$ before the \provide{},
\[
  F\geq\sqrt{1+\xi}
  \geq1+\frac{\xi}{2+\xi}=F^-.
\]
The last inequality follows from
\[
  \sqrt{1+\xi}-1
  =\frac{\xi}{\sqrt{1+\xi}+1}
  \geq\frac{\xi}{\xi+2}.
\]
Thus $z\geq Z^-$ at every step of the algorithm.

Every intermediate value of $Z^-$ is positive.  Each unresolved \reclaim{}
burns a distinct portion already included in $z_0$, so the amounts burned in
any prefix sum to at most $z_0$.  If the sum is less than $z_0$, positivity is
immediate.  If it equals $z_0$, the final \reclaim{} in that prefix could not
be valid without an earlier nonzero \provide{} increasing the token supply in
every virtual replay.  For that \provide{}, $\xi>0$ and $F^->1$, which leaves
$Z^->0$ after all $z_0$ base tokens have been subtracted.

For a final virtual pool, Lemma~\ref{lem:invariant} gives
\[
  ab=\left(\frac z\gamma\right)^2
  \geq\left(\frac{Z^-}{\gamma_0}\right)^2
  =\left(\frac{Z^-}{z_0}\right)^2a_0b_0.
\]
Also $a\leq A^+$.  Since $a(a+x)$ increases with $a>0$,
\[
  \frac{xb}{a+x}
  =\frac{xab}{a(a+x)}
  \geq
  \frac{x(Z^-/z_0)^2a_0b_0}{A^+(A^++x)}.
\]
This is~\eqref{eq:product-quote}.
\end{proof}

\subsection{An input-output balance quote}

For an A-to-B lock-swap request, let $\mathcal D$, $\mathcal P$, and
$\mathcal R$ denote the delta entries, \provides{}, and \reclaims{}.  For
$e\in\mathcal D$, write
$(\alpha_e,\beta_e)$ for its asset delta; for $e\in\mathcal P$, write
$(p_e,q_e)$ for the assets deposited; and for $e\in\mathcal R$, write $r_e$
for the amount of the token portion burned.  Define
\begin{equation}
\begin{aligned}
  U_A&=\sum_{e\in\mathcal D}\max\{0,\alpha_e\}
       +\sum_{e\in\mathcal P}p_e,\\
  D_B&=\sum_{e\in\mathcal D}\max\{0,-\beta_e\},
  &R&=\sum_{e\in\mathcal R}r_e.
\end{aligned}
  \label{eq:balance-totals}
\end{equation}
Here $U_A$ bounds additions to the input-asset balance, $D_B$ bounds
withdrawals from the output-asset balance, and $R$ is the total reclaimed
token amount.  The computation of $\ell_{\mathrm{prod}}$ ends with
$A^+=a_0+U_A$, so only $D_B$ and $R$ are new accumulators.  For a real number
$y$, write
$[y]_+=\max\{0,y\}$, and set
\begin{equation}
  \ell_{\mathrm{bal}}=
  \frac{x\left[(1-R/z_0)b_0-D_B\right]_+}
       {a_0+U_A+x}.
  \label{eq:balance-quote}
\end{equation}
The bracketed term is a lower bound on the final output-asset balance, while
the denominator uses an upper bound on the final input-asset balance.

\begin{lemma}
\label{lem:balance-bounds}
For every reachable compressed state, if $(a,b)$ is the final allocation of
one of its virtual pools, then
\[
  a\leq a_0+U_A,
  \qquad
  b\geq(1-R/z_0)b_0-D_B.
\]
\end{lemma}

\begin{proof}
List the \reclaims{} chronologically as $r_1,r_2,\ldots$, and put
$R_j=\sum_{i\leq j}r_i$.  Immediately before \reclaim{} $j$, every replay has
token supply $z_j\geq z_0-R_{j-1}$, because \provides{} only increase the token
supply.  Distinct unresolved portions are drawn from the base supply, so
$R_j\leq z_0$.  Its \reclaim{} scale therefore satisfies
\[
  1-\frac{r_j}{z_j}
  \geq
  \frac{z_0-R_j}{z_0-R_{j-1}}.
\]
The product of all \reclaim{} scales is consequently at least $1-R/z_0$.

Expand a final reserve into its base contribution and the additive
contributions of delta entries and \provides{}.  Each contribution is multiplied
by the \reclaim{} scales that follow it.  For the A reserve, omit negative
contributions and all \reclaim{} shrinkage, giving $a\leq a_0+U_A$.  For the B
reserve, omit positive contributions.  Multiplication by a number in $(0,1]$
makes a negative contribution less negative, so replacing each such
contribution by its unscaled value preserves a lower bound.  The base
coefficient is at least $1-R/z_0$, which proves the second inequality.
\end{proof}

\begin{theorem}
\label{thm:balance-quote}
For every reachable compressed state and every exact A-to-B lock-swap request
with input $x>0$, the value $\ell_{\mathrm{bal}}$ in
Equation~\eqref{eq:balance-quote} is a safe quote.
\end{theorem}

\begin{proof}
If the numerator in~\eqref{eq:balance-quote} is nonpositive, then
$\ell_{\mathrm{bal}}=0$ is safe.  Otherwise
Lemma~\ref{lem:balance-bounds} and the
monotonicity of $xb/(a+x)$ give, for every final virtual pool,
\[
  \frac{xb}{a+x}
  \geq
  \frac{x((1-R/z_0)b_0-D_B)}{a_0+U_A+x}
  =\ell_{\mathrm{bal}}.
\]
Hence $\ell_{\mathrm{bal}}\leq\mu$.
\end{proof}

\subsection{Approximation guarantees from cumulative load}
\label{sec:small-load}

Safety alone does not say how conservative a quote is.  The following
parameter gives both algorithms a useful performance guarantee.  Define
\begin{equation}
\begin{aligned}
  L_A&=\sum_{e\in\mathcal D}|\alpha_e|
       +\sum_{e\in\mathcal P}p_e,
  &L_B&=\sum_{e\in\mathcal D}|\beta_e|
       +\sum_{e\in\mathcal P}q_e,\\
  \eta&=\frac{R}{z_0}
       +\max\left\{\frac{L_A}{a_0},\frac{L_B}{b_0}\right\}.
\end{aligned}
  \label{eq:load-parameter}
\end{equation}
This dimensionless statistic measures total unresolved activity relative to
the compressed base pool.  It is conservative because it counts both
directions of every delta even though an active lock may be canceled.  It is
computed by three accumulators during the quote computation and is not a tuning
parameter of either algorithm.

For $0\leq\eta<1$, set
\begin{equation}
  c_\eta=\frac{1-\eta}{1+\eta},
  \qquad
  q_\eta(x)=
  \frac{(1-\eta)a_0+x}{(1+\eta)a_0+x}.
  \label{eq:small-load-factors}
\end{equation}

\begin{theorem}
\label{thm:small-load}
For every reachable compressed state with $\eta<1$ and every exact A-to-B
lock-swap request with input $x>0$,
\begin{equation}
\begin{aligned}
  c_\eta^2q_\eta(x)\mu&\leq\ell_{\mathrm{prod}}\leq\mu,
  &c_\eta q_\eta(x)\mu&\leq\ell_{\mathrm{bal}}\leq\mu.
\end{aligned}
  \label{eq:small-load-with-input}
\end{equation}
In particular,
\begin{equation}
  c_\eta^3\mu\leq\ell_{\mathrm{prod}}\leq\mu,
  \qquad
  c_\eta^2\mu\leq\ell_{\mathrm{bal}}\leq\mu.
  \label{eq:uniform-small-load}
\end{equation}
\end{theorem}

\begin{proof}
List the \reclaims{} chronologically as $r_1,r_2,\ldots$, and put
$R_j=\sum_{i\leq j}r_i$.  Immediately before \reclaim{} $j$, every exact replay
has token supply
\[
  z_j\geq z_0-R_{j-1},
\]
because \provides{} only increase the token supply.  Here
$R_j\leq R<z_0$, since $\eta<1$.  Its \reclaim{} scale therefore satisfies
\[
  1-\frac{r_j}{z_j}
  \geq
  \frac{z_0-R_j}{z_0-R_{j-1}}.
\]
Multiplying these inequalities shows that the product of all \reclaim{} scales
is at least $1-R/z_0$.

Expand either final asset amount into the contribution from its base amount
and the additive contributions from delta entries and \provides{}.  Each
contribution is multiplied by the \reclaim{} scales that follow it.  The base
coefficient is at least $1-R/z_0$.  A negative delta contributes no less than
its unscaled value, while a positive contribution can be omitted from a lower
bound.  For an upper bound, \reclaim{} scales and negative deltas can instead be
omitted.  Consequently, if $(a,b)$ is any final virtual-pool allocation, then
\begin{align*}
  (1-R/z_0)a_0-L_A&\leq a\leq a_0+L_A,\\
  (1-R/z_0)b_0-L_B&\leq b\leq b_0+L_B.
\end{align*}
The definition of $\eta$ gives the simpler bounds
\begin{equation}
  (1-\eta)a_0\leq a\leq(1+\eta)a_0,
  \qquad
  (1-\eta)b_0\leq b\leq(1+\eta)b_0.
  \label{eq:small-load-bounds}
\end{equation}
In particular, every virtual-pool quote value is at most
\begin{equation}
  U_\eta=\frac{x(1+\eta)b_0}{(1-\eta)a_0+x},
  \qquad\text{and hence}\qquad \mu\leq U_\eta.
  \label{eq:small-load-mu-upper}
\end{equation}

In the computation of $\ell_{\mathrm{prod}}$, every \provide{} multiplier is
at least one.  The final token lower bound and asset upper bound therefore obey
\[
  Z^-\geq z_0-R\geq(1-\eta)z_0,
  \qquad
  A^+\leq a_0+L_A\leq(1+\eta)a_0.
\]
Equation~\eqref{eq:product-quote} now gives
\[
  \ell_{\mathrm{prod}}\geq
  \frac{x(1-\eta)^2a_0b_0}
       {(1+\eta)a_0((1+\eta)a_0+x)}.
\]
Dividing this lower bound by~\eqref{eq:small-load-mu-upper} yields
$\ell_{\mathrm{prod}}/\mu\geq c_\eta^2q_\eta(x)$.

Moreover, $U_A\leq L_A$ and $D_B\leq L_B$.  The numerator of
$\ell_{\mathrm{bal}}$ is
positive when $\eta<1$, and
\[
  \ell_{\mathrm{bal}}\geq
  \frac{x(1-\eta)b_0}{(1+\eta)a_0+x},
\]
and division by~\eqref{eq:small-load-mu-upper} gives
$\ell_{\mathrm{bal}}/\mu\geq c_\eta q_\eta(x)$.  Finally,
$q_\eta(x)\geq c_\eta$ for $x\geq0$, proving
Equation~\eqref{eq:uniform-small-load}.
\end{proof}

For every fixed load cap $\bar\eta<1$, the two quotes are ordinary
constant-factor approximations on the class
$\eta\leq\bar\eta$.  Their factors are at most
\[
  \left(\frac{1+\bar\eta}{1-\bar\eta}\right)^3
  \quad\text{and}\quad
  \left(\frac{1+\bar\eta}{1-\bar\eta}\right)^2,
\]
respectively.  Equivalently,
$\ell_{\mathrm{prod}}/\mu\geq1-6\eta+O(\eta^2)$ and
$\ell_{\mathrm{bal}}/\mu\geq1-4\eta+O(\eta^2)$ as $\eta\to0$.

The following values are worst-case guarantees over every reachable event list
and every resolution of the active locks:
\[
\begin{array}{c|cc}
  \eta & c_\eta^3 & c_\eta^2\\ \hline
  0.0001 & 0.9994 & 0.9996\\
  0.001 & 0.9940 & 0.9960\\
  0.005 & 0.9704 & 0.9801\\
  0.010 & 0.9417 & 0.9607\\
  0.020 & 0.8869 & 0.9231\\
  0.050 & 0.7406 & 0.8185
\end{array}
\]
The guarantees in Equation~\eqref{eq:small-load-with-input} are stronger when
the requested input $x$ is not negligible relative to $a_0$.  Exchanging A
and B gives the corresponding statement for a B-to-A lock-swap request.

One may of course return
$\max\{\ell_{\mathrm{prod}},\ell_{\mathrm{bal}}\}$.  Both constituent quotes
are safe, so their maximum is safe.  It is strictly positive because
$\ell_{\mathrm{prod}}>0$ and is at least as accurate as either constituent
quote.

\subsection{Computation with exact rational and finite-precision arithmetic}
The input-output balance quote has a particularly small implementation using
integer arithmetic.  Clearing the token denominator gives
\begin{equation}
  \ell_{\mathrm{bal}}=
  \frac{x[\,b_0(z_0-R)-z_0D_B\,]_+}
       {z_0(a_0+U_A+x)}.
  \label{eq:integer-balance-quote}
\end{equation}
Thus widened integer sums and products compute its numerator and denominator,
and one downward-rounded integer division produces a safe lower bound on the
exact-model quote.

For rationally encoded inputs, the two algorithms are straight-line programs over
rational field operations and maximum gates, with $s=O(m+1)$ gates.  Let $h$
be the largest bit length of a numerator or denominator in the input.  A
direct induction over the displayed running sums and the recurrence for $Z^-$
bounds the numerator and denominator bit lengths of every reduced intermediate
value by some $\Lambda=O((m+1)^2(h+1))$; the final squaring in
$\ell_{\mathrm{prod}}$ and the final multiplication by $x$ preserve this bound
up to a constant factor.  Exact rational evaluation is therefore
polynomial-time in the standard Turing machine model.

More relevant for deployment, for an integer $d\geq0$, an implementation
using fixed-point arithmetic may instead enclose each
input between adjacent points of $2^{-d}\mathbb Z$, round every interval
operation outward, and return the positive part of the final lower endpoint.
Gatewise containment preserves safety.  Every positive divisor and every
positive exact-rational output of either algorithm is at least $2^{-\Lambda}$,
while all intermediate magnitudes are at most $2^\Lambda$.  A standard interval-error
induction gives a universal constant
$C$ such that, for every integer $\kappa\geq1$, choosing
$d\geq C(s\Lambda+\Lambda+\kappa)$ and using checked integer arithmetic with polynomial
total word width suffices to return, for
$\ell\in\{\ell_{\mathrm{prod}},\ell_{\mathrm{bal}}\}$, a value
$\widehat\ell$ satisfying
\[
  (1-2^{-\kappa})\ell\leq\widehat\ell\leq\ell\leq\mu.
\]
At this precision, the rounded quote incurs relative loss at most
$2^{-\kappa}$.  At lower precision, a divisor interval containing zero
triggers the safe fallback value zero.

\section{Concluding Remarks}
\label{sec:conclusion}

The open problem of Aanes et al.~\cite{aanes2025lockswap} asked whether the
exact worst-case quote of a retroactive pool can be computed efficiently in
general.  We have shown that, unless $\mathrm{P}=\mathrm{NP}$, neither the
exact value nor any fixed-factor safe approximation can be computed in
polynomial time.  Fortunately, we have also shown that this hardness
is not an obstacle to deployment: a retroactive pool does not need the exact
minimum, only a safe quote, and safe quotes with an explicit, per-state
accuracy certificate can be computed in practice with one linear scan and a constant number of registers.  In the regime where the stored unresolved activity is light
relative to the pool, which is the regime in which one would expect a healthy pool to
operate, the certified accuracy is close to optimal. If the pool operates outside of that regime, the quotes are still safe, though they might not be
very good.
The
safety arguments use only the inequality $(a+u)(b-v)\geq ab$ for executed
locks, which fee-bearing execution preserves, so both algorithms remain safe in a
pool that charges fees; we can expect the load-dependent guarantees to extend as
well, suitably modified, but we do not pursue this here.

\bibliographystyle{alpha}
\bibliography{references}

@misc{aanes2025lockswap,
  author = {Aanes, Jon Michael and Gravgaard, Jesper Balman and Miltersen, Peter Bro and Nielsen, Kurt and Pourpouneh, Mohsen},
  title = {Automated Market Makers for Cross-chain {DeFi} and Sharded Blockchains},
  year = {2025},
  eprint = {2309.14290},
  archivePrefix = {arXiv},
  primaryClass = {cs.DC},
  doi = {10.48550/arXiv.2309.14290},
  howpublished = {\href{https://arxiv.org/abs/2309.14290v3}{arXiv:2309.14290v3 [cs.DC]}},
  note = {Version 3, last revised January 27, 2025}
}

@book{gareyjohnson1979,
  author = {Garey, Michael R. and Johnson, David S.},
  title = {Computers and Intractability: A Guide to the Theory of NP-Completeness},
  publisher = {W. H. Freeman},
  address = {San Francisco},
  year = {1979},
  isbn = {0-7167-1045-5}
}

\end{document}